\documentclass[letterpaper]{article} 
\usepackage{aaai2026}  
\usepackage{times}  
\usepackage{helvet}  
\usepackage{courier}  
\usepackage[hyphens]{url}  
\usepackage{graphicx} 
\usepackage{natbib}  
\usepackage{caption} 
\newcommand{\los}{\mathcal{L}}
\newcommand{\watchers}{\mathcal{W}}
\newcommand{\unseen}{\mathcal{U}}
\newcommand{\freecells}{\mathcal{C}}
\newcommand{\nodeunseen}{R}

\newcommand{\pathunion}{\mathcal{P}}
\newcommand{\neighbors}{\mathcal{N}}
\newcommand{\starts}{\mathcal{S}}
\newcommand{\paths}{\boldsymbol{\pi}}
\newcommand{\astar}{A\textsuperscript{*}\xspace}
\newcommand{\algname}{MWRP-CP\textsuperscript{3}\xspace}
\newcommand{\objective}{\mathcal{O}}
\newcommand{\subalgname}{MxW\astar}
\newcommand{\anytimesubalgname}{A\subalgname}

\usepackage{layouts}
\usepackage{amsthm}
\usepackage{amsmath}
\usepackage{bm}
\usepackage[english]{babel}
\usepackage[algo2e,ruled,vlined,linesnumbered,noend]{algorithm2e} 
\usepackage{graphicx}
\usepackage{multirow}
\usepackage{makecell}
\usepackage{array}
\usepackage{subcaption}   
\newtheorem{theorem}{Theorem}

\SetAlgoNlRelativeSize{-1}   
\SetNlSty{}{}{:}

\usepackage{tikz}
\usetikzlibrary{positioning, arrows.meta} 

\usepackage{algorithm}
\usepackage[noend]{algorithmic}

\usepackage{newfloat}
\usepackage{listings}
\DeclareCaptionStyle{ruled}{labelfont=normalfont,labelsep=colon,strut=off} 
\floatstyle{ruled}
\newfloat{listing}{tb}{lst}{}
\floatname{listing}{Listing}
\title{Scalable Algorithms with Provable Optimality Bounds for the Multiple Watchman Route Problem}
\author {
    Srikar Gouru\textsuperscript{\rm 1},
    Ariel Felner\textsuperscript{\rm 2},
    Jiaoyang Li\textsuperscript{\rm 1}
}
\affiliations {
    \textsuperscript{\rm 1}Carnegie Mellon University\\
    \textsuperscript{\rm 2}Ben-Gurion University of the Negev\\
    srikargouru@cmu.edu, felner@bgu.ac.il, jiaoyangli@cmu.edu
}

\begin{document}

\maketitle

\begin{abstract}
In this paper, we tackle the Multiple Watchman Route Problem (MWRP), which aims to find a set of paths that $M$ watchmen can follow such that every location on the map can be seen by at least one watchman. First, we propose multiple methods to reduce the state space over which a search needs to be conducted by pruning map areas that are guaranteed to be seen en route to other areas. Next, we introduce \algname, an efficient optimal planner that combines these methods with techniques that improve the quality and calculation time of existing heuristics. We present several suboptimal algorithms with bounds on solution quality, including \subalgname, a general variant of weighted \astar for makespan problems. We also present anytime variations of our suboptimal algorithms, as well as techniques to improve an existing suboptimal solution by solving multiple decomposed sub-problems. We show that \algname can reduce the search space by more than 95\% and runs more than 200$\times$ faster than existing optimal algorithms on 2D grid maps. We also show that our suboptimal algorithms solve maps 3$\times$ larger than those solvable by \algname. See \verb|mwrp-cp3.github.io| for the open source codebase and video demonstrations.
\end{abstract}


\section{Introduction}
\label{sec:introduction}
Coverage Path Planning (CPP) is the task of computing a path that passes over every point in a given environment. CPP plays a vital role in emergency scenarios such as searching for survivors in disaster relief \cite{cho2021coverage} and locating active wildfires \cite{bezas2022coverage}. Extensive research has been conducted for exploration in \textit{unknown} environments including frontier-based \cite{yamauchi1997frontier} and sampling-based \cite{steinbrink2021rapidly} algorithms. For CPP in \textit{known} environments, methods such as cellular decomposition \cite{huang2001optimal} and genetic algorithms \cite{jimenez2007optimal} have been studied to compute optimal paths. However, these methods are designed for robots with a pre-determined coverage radius, such as lawnmowers and vacuums, rather than robots whose coverage may instead depend on visibility or line-of-sight.

Specifically, we focus on the \textit{Watchman Route Problem} (WRP) which requires a traveling agent, the \textit{watchman}, to find the shortest path such that every point on the environment can be seen from at least one location on the path. WRP was first introduced by Chin and Ntafos, who presented an efficient solution for rectilinear polygons and proved that the problem was NP-hard in polygons with holes \cite{chin1986optimum}. Since then, many problem variations have been introduced. The $k$-Watchman Route Problem ($k$-WRP) \cite{packer2008computing} involves $k$ traveling watchmen whose combined vision must encompass the entire environment, with the goal of minimizing either the longest watchman's path (\textit{min-max / makespan}) or the total path length (\textit{min-sum}). Computing optimal paths for an arbitrary $k$ watchmen was proven to be NP-hard for both the makespan \cite{mitchell1991watchman} and min-sum \cite{nilsson1996guarding} objectives in simple polygons. Moreover, WRP can be either \textit{anchored}, in which the watchmen have predetermined starting locations, or \textit{floating}, in which the watchman paths can start anywhere.

Recently, a discrete variation of single-watchman WRP was introduced for arbitrary graphs, with a focus on 2D grid maps. Unlike the classic WRP and $k$-WRP problems, graph-based WRP works on an arbitrary \textit{Line-Of-Sight} (LOS) function that determines which cells are visible from a given cell \cite{seiref2020solving}. Additionally, graph-based WRP makes no prior assumptions on the environment geometry and the watchman's LOS, allowing for complex problem instances to be modeled. The problem was proven to be NP-hard to solve optimally. Suboptimal algorithms have also been proposed \cite{yaffe2021suboptimally}.

Graph-based WRP has since been extended to the graph-based \textit{Multiple Watchman Route Problem} (MWRP), in which every cell must be visible from at least one cell along at least one of the watchman paths \cite{livne2023optimally}. A new variant of \astar called MWRP-\astar was introduced, which performs a joint-space heuristic search to compute an optimal set of paths for both the min-sum and makespan objectives. They were able to solve grids of up to 200 free cells and up to 6 agents in several seconds. However, MWRP-\astar is infeasible for real-world scenarios where maps consist of thousands of cells. We focus on the makespan objective since most exploration problems are time-critical due to some emergency (i.e., disaster relief, firefighting). Our contributions are summarized as follows:

\begin{enumerate}
    \item \textbf{Optimal Search Algorithm:} We present \algname, which consists of methods for reducing the state space of the algorithm by determining that some cells in the map are strictly harder to see than others. Additionally, \algname includes algorithmic enhancements that significantly improve heuristic quality and calculation time.
    \item \textbf{Bounded Suboptimal Search (BSS):} We introduce \textit{Minimax Weighted \astar} (\subalgname), a general variant of weighted \astar \cite{pohl1970heuristic} that efficiently computes bounded-quality solutions for makespan problems. Additionally, we introduce two heuristics that are used in conjunction with Focal Search \cite{pearl1982studies} to solve the MWRP problem with makespan. We also formulate anytime variations for both algorithms.
    \item \textbf{Postprocessing Framework:} We present a framework that efficiently improves the makespan of an existing suboptimal solution by partitioning the map into and solving fast, decomposed sub-problems for each agent.
\end{enumerate}

We show that \algname can reduce the search space by more than 95\% on complex maps and computes optimal paths more than 200$\times$ faster than existing baselines. We also show that our suboptimal algorithms scale in terms of map size (1,500+ grid cells) and agent count (5+ agents).

\section{Background}\label{sec:background}
The WRP and MWRP problems can be represented as any graph structure \cite{seiref2020solving, livne2023optimally}, and all of our techniques and algorithms are generalizable to arbitrary graphs. For simplicity, we formulate the problem as a 2D grid graph composed of free and obstacle cells, a common approach in coverage-related problems.

\subsection{Problem Formulation}

The MWRP problem is defined by the tuple $(M, \freecells, \unseen, \starts, \neighbors, \los, \objective)$, where $M$ is the number of agents, $\freecells$ is the set of free cells on the grid that the agents can be at, $\unseen \subseteq \freecells$ is the set of cells that need to be seen by at least one of the agents, and $\starts = \{\starts_1, ..., \starts_M\}$, where $\starts_k \in \freecells$ is the starting location of agent $a_k$. In general, $\unseen = \freecells$, however we separate these terms because we introduce certain subproblems where not every cell in $\freecells$ needs to be explicitly seen. The neighbor function, $\neighbors$, defines the cells that an agent can move to in a single action from a given cell, and the LOS function, $\los$, defines all cells visible to an agent from a given cell. Figure~\ref{fig:example_problem} (left) visualizes \textit{four-way movement}, an example neighbor function, and \textit{Bresenham LOS}, an example LOS function commonly used in computer graphics that simulates a continuous field-of-view in discrete maps \cite{5388473}. We define the \textit{watchers} function, $\watchers$, as the inverse of the LOS function, such that $s_2 \in \los(s_1) \leftrightarrow s_1 \in \watchers(s_2) \, \forall \, s_1,s_2 \in \freecells$. For many LOS functions, including Bresenham LOS, $\watchers$ may not be symmetric, meaning there is no guarantee that $\watchers(s) = \los(s)$. We define $d(s_1, s_2)$ as the length of the shortest traversable path between $s_1$ and $s_2$.

A solution to the MWRP problem must be a set of paths $\paths = \{\pi_1, ..., \pi_M\}$, where each path $\pi_k$ is a list of cells that agent $a_k$ traverses starting at $\starts_k$ and following the movement constraints defined by $\neighbors$. A set of paths, $\paths$, is a solution to the problem if every cell in $\unseen$ is seen during at least one agent's path. That is, defining the function $\pathunion(\paths) = \bigcup \, \pi_k$, $\paths$ is a solution if $\watchers(s) \, \cap \pathunion(\paths) \neq \emptyset$ for every cell $s \in \unseen$. The cost of a path is its length, defined as $c(\pi_k) = |\pi_k|$. $\objective$ is the \textit{makespan} objective function, defined as $\objective(\paths) = \max\limits_{\pi_k \in \paths}c(\pi_k)$.

\begin{figure}[t]
    \centering

    \includegraphics[width=\linewidth]{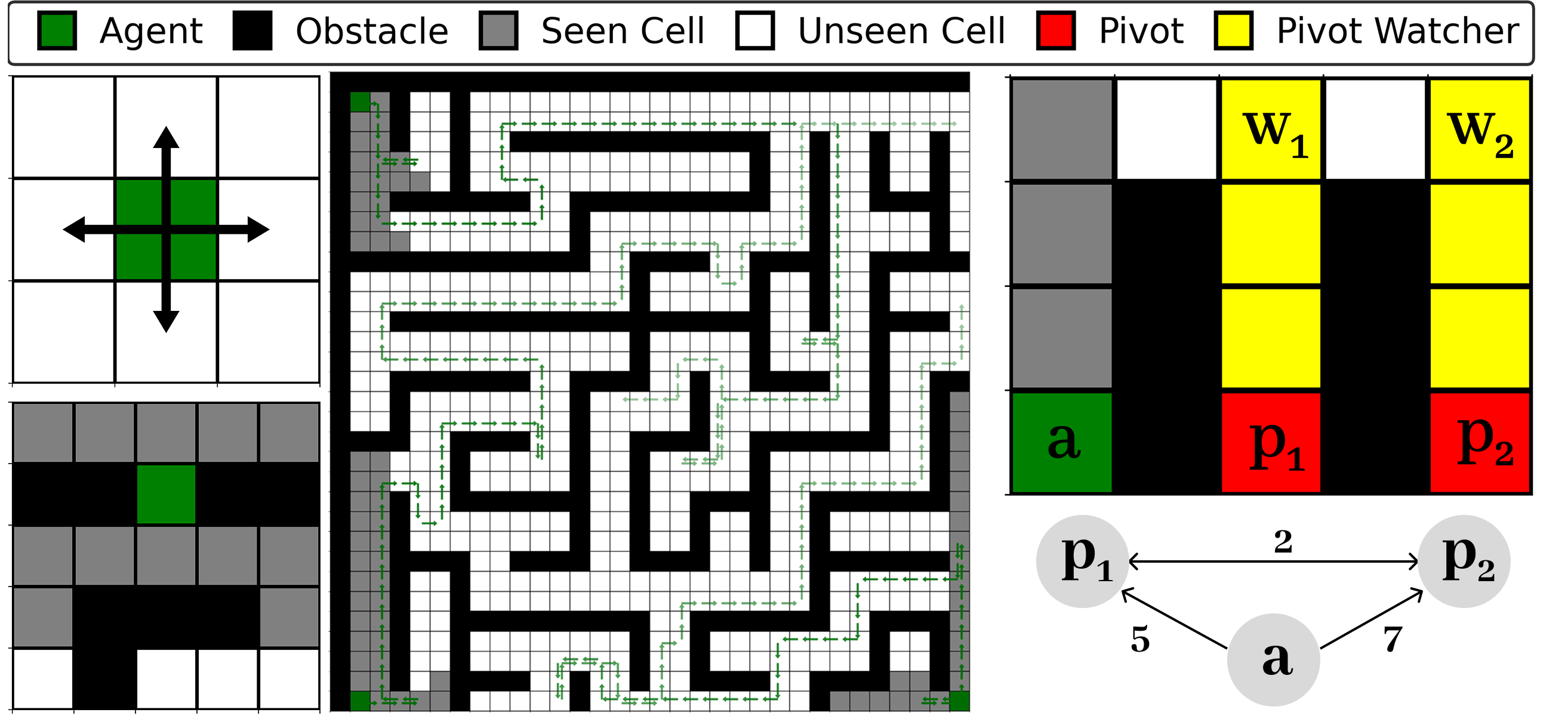}

    \caption{Examples of four-way movement (top left), Bresenham LOS (bottom left), an MWRP problem with the optimal makespan solution (middle), an example scenario for the mTSP heuristic (top right), and its corresponding $G_\mathit{DLS}$ (bottom right). The legend applies to all figures in this paper.}
    \label{fig:example_problem}
\end{figure}

Because exploration problems generally require few agents searching over a large area, we assume that each cell $s \in \freecells$ can hold all $M$ agents, and thus we do not consider agent-agent collisions for our solution. This assumption was also used in previous work \cite{livne2023optimally}. The problem is solved once offline, and we assume perfect information regarding the map structure and agent starting locations. Figure~\ref{fig:example_problem} (middle) shows an example problem and a makespan-optimal solution on a 32$\times$32 Maze map with 3 agents.

\subsection{MWRP-A* Algorithm}

\label{subsec: mwrp_algorithm}

In this section we summarize the MWRP-\astar algorithm by Livne et al., a variation of joint-space \astar that optimally solves the MWRP problem \cite{livne2023optimally}. \textbf{Node}: Each node in the search tree consists of a list of tuples, $(s_k, c_k)$, representing the current location and the cumulative cost for agent $a_k$. Additionally, each node keeps track of its \textit{residual set} $\nodeunseen \subseteq \unseen$, which is the set of residual (or remaining) cells that need to be seen. The cost of the node is $\max\limits_k\{c_k\}$. \textbf{Root node}: The tuples are $(\starts_k, 0)$ for each agent $a_k$, and the residual set, $\nodeunseen = \unseen \ \setminus \bigcup\limits \los(S_k)$, represents the residual cells not visible from any agent's starting location. \textbf{Goal node:} Any node where $\nodeunseen = \emptyset$ is considered a goal node.

\subsubsection{Successors}

Successor nodes are generated using a joint-space expansion. To avoid generating nodes where agents do not see any new cells, MWRP-\astar utilizes the \textit{Expanding Borders} mechanism, which finds nearby cells for each agent that have LOS to a previously unseen residual cell. These new cells, along with the cumulative cost to reach them, are the successor states of the agent. Each agent also has the option to \textit{terminate}, which removes it from the set of agents in the successor node. The successor nodes are the Cartesian product of the successor states for each agent. In each successor node, $\nodeunseen$ is computed by copying $\nodeunseen$ from the parent node and removing cells that the agents have LOS to from their new locations.

\subsubsection{Node Dominance}

To avoid unnecessary re-expansions, \textit{dominated} nodes are pruned out during the search. Node $n$ \textit{dominates} node $n'$ if (1) all of their agent locations are equal ($n.s_k = n'.s_k$), (2) the residual set for node $n$ is a subset of the residual set for node $n'$ ($n.\nodeunseen \subseteq n'.\nodeunseen$), and (3) the cumulative cost for each agent in node $n$ is less than or equal to the cumulative cost in node $n'$ ($n.c_k \leq n'.c_k$).

\subsubsection{Lazy Heuristic Evaluation}

Nodes are initially inserted into the \textit{OPEN} list with their $f$-value calculated using the \textit{Singleton heuristic}. If a node is popped from \textit{OPEN} for the first time, a more informed $f$-value is calculated using the \textit{multiple Traveling Salesman Problem (mTSP) heuristic}, and the node is re-inserted into \textit{OPEN} with its new $f$-value. When a node is popped from \textit{OPEN} for the second time, it is expanded and successors are generated.

\subsubsection{Singleton Heuristic}

The Singleton heuristic defines $f_s = \min\limits_{k} \{c_k + w_{s,k}\}$ as the lowest agent cost at which cell $s$ can be seen, where $w_{s,k}$ is the minimum cost for agent $a_k$ to reach a watcher of $s$. Since every cell in $\nodeunseen$ must be seen, the admissible Singleton $f$-value is $\max\{f_s\}$ among all $s \in \nodeunseen$.

\subsubsection{mTSP Heuristic}

First, a set of \textit{pivots} is selected from $\nodeunseen$ such that no two pivots share any watchers. Next, a graph $G_\mathit{DLS}$ is constructed, where each vertex $v_p$ in $G_\mathit{DLS}$ corresponds to pivot $p$, and each vertex $v_a$ corresponds to agent $a$. Since any solution requires all pivots to be seen by at least one agent, the heuristic is computed by solving mTSP on the graph $G_\mathit{DLS}$, where each vertex $v_p$ must be reached by at least one agent, with the agents starting at $v_a$. For the heuristic to be admissible, edges in $G_\mathit{DLS}$ are calculated with underestimated costs. Edge costs from agent vertices $v_a$ to pivot vertices $v_p$ are calculated as the minimum distance between agent $a$ and any watcher of pivot $w \in \watchers(p)$. Edge costs from pivot vertices $v_{p_i}$ to other pivot vertices $v_{p_j}$ are calculated as the minimum distance between any watcher $w_i \in \watchers(p_i)$ and any watcher $w_j \in \watchers(p_j)$, as a lower bound on the travel distance to go from seeing pivot $p_i$ to pivot $p_j$.

Figure~\ref{fig:example_problem} (top right) shows an example search node with a single agent $a$, pivots $p_1$ and $p_2$, and the corresponding $G_\mathit{DLS}$ (bottom right). The closest distance from agent $a$ to a watcher of $p_1$ (labeled $w_1$) is 5, and the closest distance from agent $a$ to a watcher of $p_2$ (labeled $w_2$) is 7. Additionally, the closest distance from any watcher of $p_1$ to any watcher of $p_2$ is 2 (the distance from $w_1$ to $w_2$), thus resulting in the $G_\mathit{DLS}$ graph shown in Figure~\ref{fig:example_problem} (bottom right).

Defining $h_k$ as the path length of each agent $a_k$ on $G_\mathit{DLS}$, solving mTSP to minimize $\min\{c_k + h_k\}$ \cite{gavishmtsp} results in an admissible $f$-value. However, this $f$-value is not consistent, so we utilize the \textit{pathmax} technique to ensure consistency by setting $f(n) = \max\{f(n), f(n_p)\}$, where $n_p$ is node $n$'s parent during the \astar search \cite{mero1984heuristic}.

\section{\algname}
\label{sec:algorithmic_improvements}
This section describes our optimal search algorithm, \algname, which is MWRP-\astar with Cell and Path dominance, Pivot pruning, and Parallel heuristic calculation. We first introduce cell and path dominance, two methods for reducing the algorithm's state space. Then, we introduce pivot pruning and parallel heuristic calculation, two enhancements that directly speed up the original MWRP-\astar algorithm.

\subsection{State Space Reduction}

Each node in the search space is uniquely defined by the location of each agent, which is a cell in $\freecells$, and the residual set $\nodeunseen$, which is a subset of $\unseen$. Thus, the state space is $O(|\freecells|^M2^{|\unseen|})$. We introduce \textit{cell} and \textit{path dominance}, two methods to reduce $\unseen$ by identifying cells that can be ignored during the search, a process that we call \textit{cell pruning}. We prove that both methods do not alter the set of possible solutions, thus preserving the algorithm's optimality guarantees.

\subsubsection{Cell Dominance}

\captionsetup[subfigure]{skip=0pt}   

\begin{figure}[t]
    \centering

    \begin{subfigure}{0.32\linewidth}
        \centering
        \includegraphics[width=\linewidth]{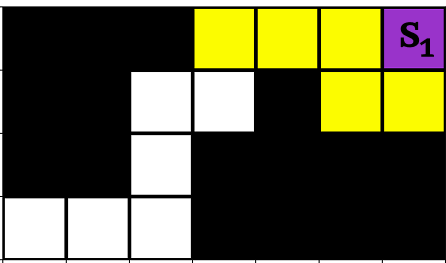}
        \caption{}
    \end{subfigure}
    \hfill
    \begin{subfigure}{0.32\linewidth}
        \centering
        \includegraphics[width=\linewidth]{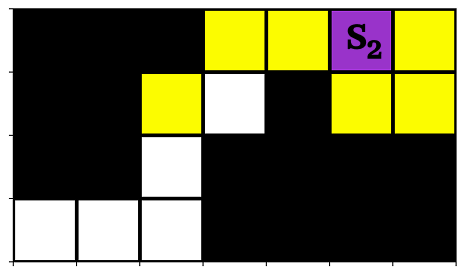}
        \caption{}
    \end{subfigure}
    \hfill
    \begin{subfigure}{0.32\linewidth}
        \centering
        \includegraphics[width=\linewidth]{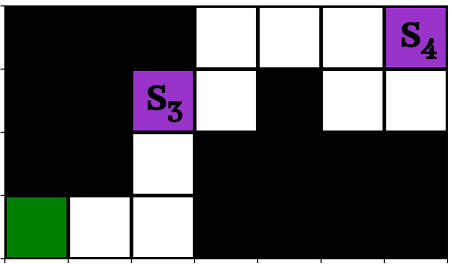}
        \caption{}
    \end{subfigure}

    \caption{Figures(a) and (b) show the watchers (yellow) of cells $s_1$ and $s_2$ (purple), respectively. The cells are also watchers of themselves. Figure(c) shows an example with an agent (green) where $s_3$ is along the agent's path to $s_4$.}
    \label{fig:cell_domination_example}
\end{figure}

The first method, \textit{cell dominance} (CD), stems from the observation that some cells are guaranteed to be seen when an agent sees other cells, no matter where the agent is. Intuitively, looking at the example shown in Figure~\ref{fig:cell_domination_example}(a) and (b), it is impossible to see $s_1$ without also seeing $s_2$. Formally, this is because the watchers of $s_1$ are a subset of the watchers of $s_2$. Thus, no matter what watcher of $s_1$ an agent sees the cell from, it simultaneously sees $s_2$. We refer to this property as cell dominance, specifically that $s_1$ \textit{dominates} $s_2$. During our search, we ignore $s_2$, since seeing $s_1$ must also result in seeing $s_2$. Theorem~\ref{thm: cell_dominance} proves that cell dominance does not sacrifice optimality.

\begin{theorem}
\label{thm: cell_dominance}
Given a cell $s_j \in \unseen$, if there exists another cell $s_i \in \unseen$ such that $\watchers (s_i) \subseteq \watchers (s_j)$, then any solution $\paths$ that sees $\unseen \setminus \{s_j\}$ also sees $\unseen$.
\end{theorem}

\begin{proof}
Let $\paths$ be any solution that sees every cell in $\unseen \, \setminus \, \{s_j\}$. Since $s_i \in \unseen \setminus \{s_j\}$, the solution must include a watcher of $s_i$, so there exists a cell $k \in  \pathunion(\paths) \, \cap \, \watchers(s_i)$. Since $\watchers(s_i) \subseteq \watchers(s_j)$, it must be that $k \in \watchers(s_j)$, thus $\paths$ sees $s_j$.
\end{proof}

\begin{algorithm}[t]
\caption{Cell Dominance (CD)}
\label{alg:cell_dominance}
\SetKw{Continue}{continue}
\KwIn{$\watchers, \unseen$}
\KwOut{Updated $\unseen$ with dominated cells pruned}
\begin{algorithmic}[1] 
\FOR{$s_i \in \unseen$}\label{cd: first_for_loop}
\FOR{$s_j \in \unseen \setminus \{s_i\}$} \label{cd: second_for_loop}
\STATE \lIf{$\watchers (s_i) \subseteq \watchers(s_j)$} { \label{cd: watcher_comparison}
            $\unseen = \unseen \setminus \{s_j\}$}

\ENDFOR
\ENDFOR
\STATE \textbf{return} $\unseen$;
\end{algorithmic}
\end{algorithm}

In Algorithm~\ref{alg:cell_dominance}, we loop through every pair of unseen cells (lines \ref{cd: first_for_loop}-\ref{cd: second_for_loop}) and prune one of them if it is dominated by the other (line \ref{cd: watcher_comparison}). Since each cell has at most $|\freecells|$ watchers, the runtime complexity of CD is $O(|\freecells||\unseen|^2)$. CD is executed once per problem instance to compute the updated $\unseen$ prior to running the \astar algorithm discussed in Section~\ref{subsec: mwrp_algorithm}.

\subsubsection{Path Dominance}

The second method, \textit{path dominance} (PD), stems from the idea that given the starting locations of the agents, the agents are guaranteed to see some cells on the way to seeing other cells. In the example shown in Figure~\ref{fig:cell_domination_example}(c), $s_3$ lies on the agent's path to seeing $s_4$, and it is impossible to see $s_4$ without seeing $s_3$ first. Thus, for any pair of distinct cells $s_i$ and $s_j$ in $\unseen$, if all the paths to seeing $s_i$ require the agent to also see $s_j$, then we can safely ignore $s_j$ during our search, since any solution that sees $s_i$ is guaranteed to also see $s_j$. We call this behavior \textit{path dominance}, specifically that $s_i$ \textit{dominates} $s_j$. For a cell $s_j \in \unseen$, we determine if another cell $s_i \in \unseen$ dominates it by performing a multi-source breadth-first search (BFS), starting at $\starts$, to find all cells $\mathbf{\Pi}$ that can be reached by at least one agent without seeing $s_j$. If there is some cell $s_i$ such that no watcher of $s_i$ is in $\mathbf{\Pi}$, then $s_i$ \textit{dominates} $s_j$. Theorem~\ref{thm: path_dominance} proves that path dominance does not sacrifice optimality.

\begin{theorem}
\label{thm: path_dominance}
Given a cell $s_j \in \unseen$, we define $\Pi$ as the set of all paths that originate at the starting location of an agent and do not see $s_j$ from any cell along the path. We also define $\mathbf{\Pi} = \bigcup\limits_{\pi_i \in \Pi}\pi_i$ as the set of cells that is in at least one path in $\Pi$. If there exists a cell $s_i \neq s_j$ such that $\watchers(s_i) \, \cap \, \mathbf{\Pi} = \emptyset$, then any solution $\paths$ that sees $\unseen \setminus \{s_j\}$ also sees $\unseen$.
\end{theorem}

\begin{proof}
Let $\paths$ be any solution that sees every cell in $\unseen \setminus \{s_j\}$. Since $s_i \in \unseen \setminus \{s_j\}$, there exists a watcher $k \in \watchers(s_i)$ that was visited by at least one agent, so $k \in \pi_a$ for some $\pi_a \in \paths$. Since $\watchers(s_i) \, \cap \, \mathbf{\Pi} = \emptyset$, we know $k \notin \mathbf{\Pi}$. This means that $\pi_a \notin \Pi$, so $\pi_a$ must have seen $s_j$ along its path.
\end{proof}

\begin{algorithm}[t]
\caption{subgraphBFS}
\label{alg:subgraphbfs}
\SetKw{Break}{break}
\SetKw{Goto}{goto}
\SetKw{Continue}{continue}
\KwIn{$\starts, \neighbors, \freecells'$}
\KwOut{$\mathbf{\Pi}$ containing all reachable cells within $\freecells'$}    

\begin{algorithmic}[1] 
\STATE $Q = \starts$;\;

\STATE $\mathbf{\Pi} = \emptyset$;\;

\WHILE{$Q \neq \emptyset$}
    \STATE $s = Q.pop()$;\;

    \STATE \lIf{$s \in \mathbf{\Pi}$ \textbf{or} $s \notin \freecells'$} {\label{bfs: freecells_check}\Continue}

    \STATE $\mathbf{\Pi} = \mathbf{\Pi} \cup s$;\;

    \STATE \lFor{$n \in \neighbors(s)$} {$Q = Q \cup n$}        

\ENDWHILE
\STATE \textbf{return} $\mathbf{\Pi}$;
\end{algorithmic}
\end{algorithm}

\begin{algorithm}[t]
\setcounter{AlgoLine}{0}
\caption{Path Dominance (PD)}
\label{alg:path_dominance}
\SetKw{Break}{break}
\SetKw{Goto}{goto}
\KwIn{$\watchers, \neighbors, \unseen, \starts$}
\KwOut{Updated $\unseen$ with dominated cells pruned}
\begin{algorithmic}[1] 

    
    
            

\FOR{$s_j \in \unseen$} \label{pd: first_loop}
    \STATE $\freecells' = \freecells \setminus \watchers(s_j)$;\; \label{pd: subset_creation}
    
    \STATE $\mathbf{\Pi} =$ \texttt{subgraphBFS(}$\starts, \neighbors, \freecells'$\texttt{)};\; \label{pd: bfs_call}
    
    \FOR{$s_i \in \unseen \setminus \{s_j\}$} \label{pd: second_loop}
        \IF{$\watchers(s_i) \cap \mathbf{\Pi} = \emptyset$} \label{pd: intersection_check}
            \STATE $\unseen = \unseen \setminus \{s_j\}$;\; \label{pd: unseen_pruning} 
            
            \STATE \Break;\;
        \ENDIF
    \ENDFOR
\ENDFOR
\STATE \textbf{return} $\unseen$;
\end{algorithmic}
\end{algorithm}

First, we introduce \texttt{subgraphBFS} (Algorithm~\ref{alg:subgraphbfs}), a modified floodfill algorithm that only expands cells within a given cell set $\freecells'$ (line \ref{bfs: freecells_check}). The method returns all cells in $\freecells'$ reachable from $\starts$ and has a runtime complexity of $O(|\freecells'|)$.

In Algorithm~\ref{alg:path_dominance} we describe the PD algorithm, in which we loop through all cells $s_j \in \unseen$ to determine if $s_j$ is path dominated by any other cell $s_i$. First, we compute $\freecells'$, the subset of $\freecells$ with all watchers of $s_j$ removed (line \ref{pd: subset_creation}). Then we call \texttt{subgraphBFS} on $\freecells'$ to compute $\mathbf{\Pi}$, the set of cells that an agent can reach while avoiding LOS to $s_j$ (line \ref{pd: bfs_call}). If there is a cell $s_i$ that has no watchers in $\mathbf{\Pi}$, then $s_i$ dominates $s_j$, so $s_j$ is pruned (lines \ref{pd: second_loop}-\ref{pd: unseen_pruning}). The runtime complexity for this check is $O(|\freecells||\unseen|)$, which dominates the $O(|\freecells|)$ runtime of the \textit{subgraphBFS} algorithm. Thus, like CD, the runtime complexity of the PD algorithm is $O(|\freecells||\unseen|^2)$. PD is run once per problem instance to compute the updated $\unseen$ prior to running the \astar algorithm discussed in Section~\ref{subsec: mwrp_algorithm}.

\begin{figure}[t]
    \centering

    \includegraphics[width=\linewidth]{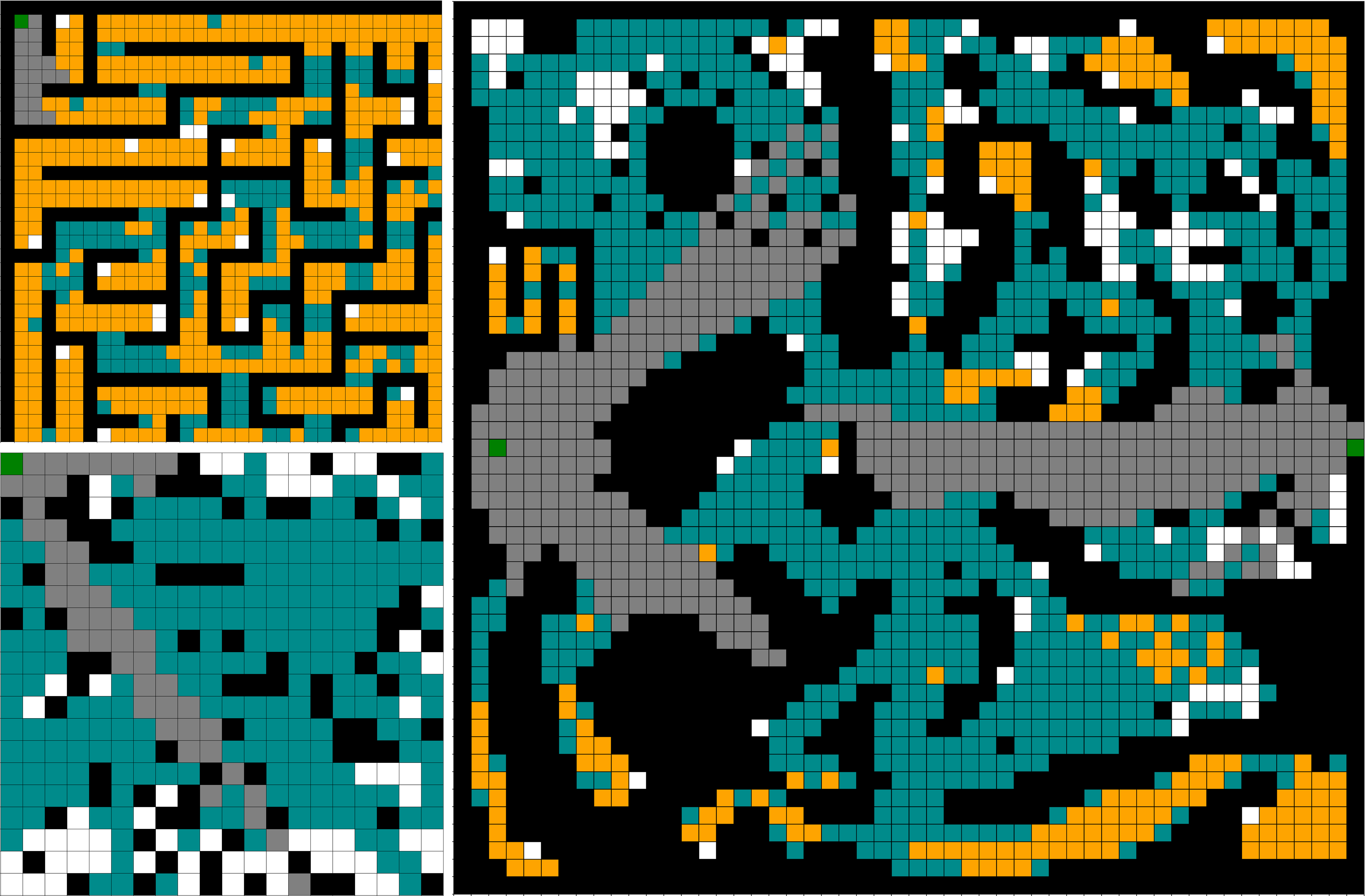}

    \caption{Comparison of state space reduction showing cells removed by \textit{cell dominance} (orange only) and \textit{path dominance} (orange + blue) on a Maze map (top left) \cite{stern2019mapf}, a randomized map (bottom left), and a Minecraft-inspired map (right) \cite{chong2024optimal}.}
    \label{fig:cell_path_domination_comparison}
\end{figure}

We also note that the set of cells pruned by CD is a subset of those pruned by PD. If $s_i$ dominates $s_j$ in CD, then there exists no cell from which we can see $s_i$ without also seeing $s_j$, thus $s_i$ dominates $s_j$ in PD as well. Figure~\ref{fig:cell_path_domination_comparison} visualizes the cells dominated by CD and PD. Our experiments show that CD generally runs faster than PD due to time-intensive BFS computation in PD. Thus, we first run CD followed by PD, reducing the cells that PD needs to loop through. We call this method cell and path dominance (CPD).

\subsection{Pivot Pruning}

\begin{figure}[t]
    \centering

    \begin{subfigure}{0.38\linewidth}
        \centering
        \includegraphics[width=\linewidth]{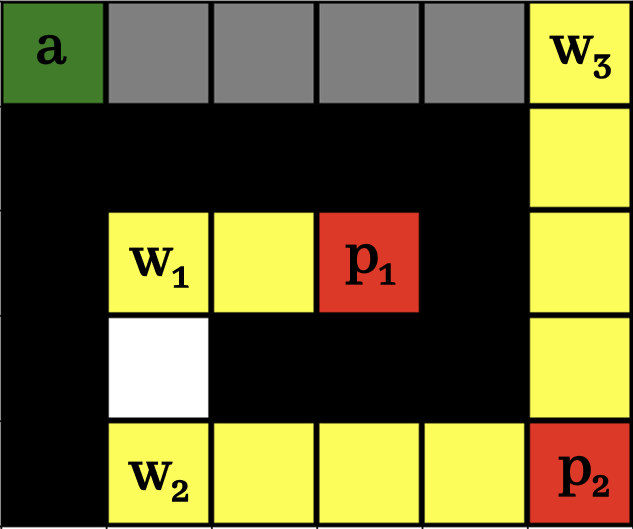}
        \caption{}
    \end{subfigure}
    \hfill
    \begin{subfigure}{0.2\linewidth}
        \centering
        \includegraphics[width=\linewidth]{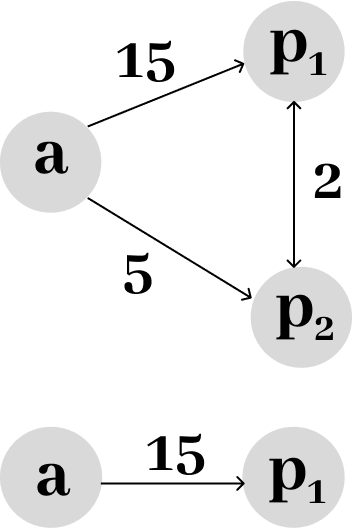}
        \caption{}
    \end{subfigure}
    \hfill
    \begin{subfigure}{0.38\linewidth}
        \centering
        \includegraphics[width=\linewidth]{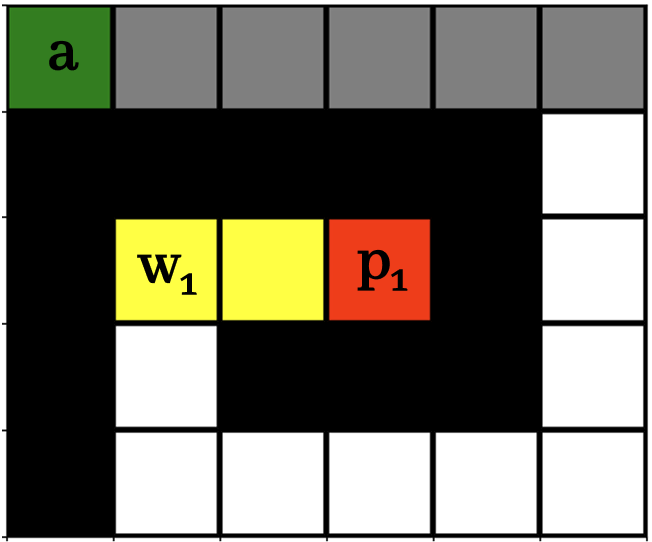}
        \caption{}
    \end{subfigure}

    \begin{subfigure}{0.38\linewidth}
        \centering
        \includegraphics[width=\linewidth]{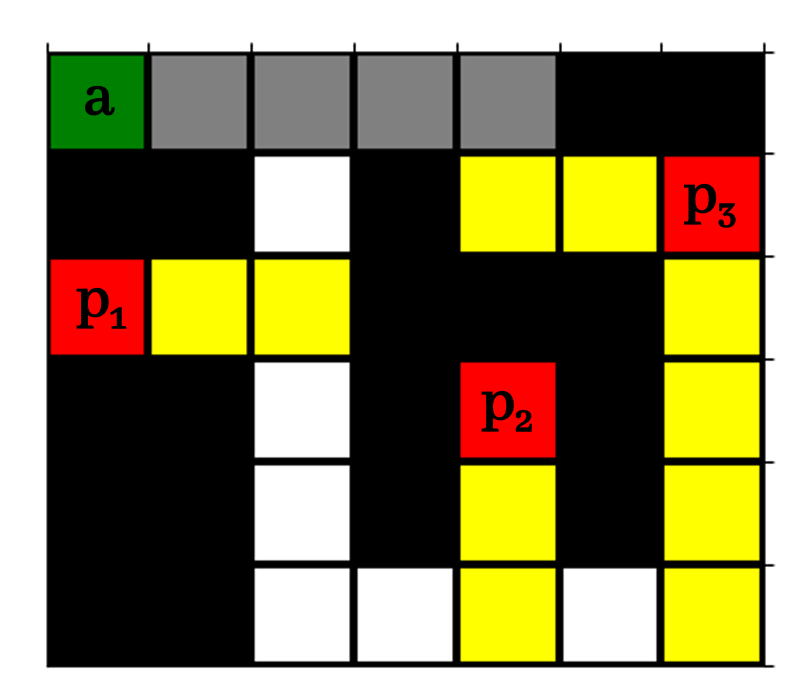}
        \caption{}
    \end{subfigure}
    \hfill
    \begin{subfigure}{0.2\linewidth}
        \centering
        \includegraphics[width=\linewidth]{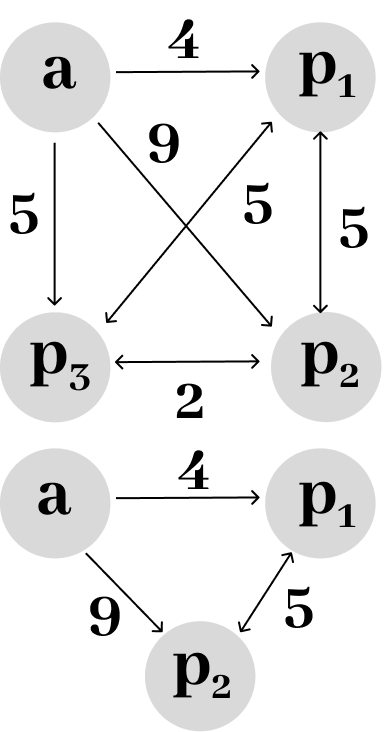}
        \caption{}
    \end{subfigure}
    \hfill
    \begin{subfigure}{0.38\linewidth}
        \centering
        \includegraphics[width=\linewidth]{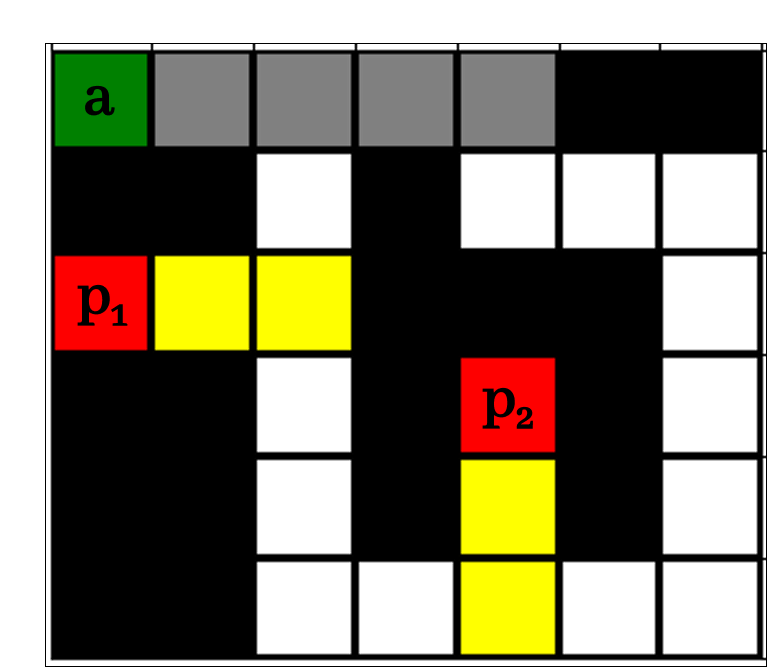}
        \caption{}
    \end{subfigure}

    \caption{Examples where pivot pruning improves (top) and worsens (bottom) the heuristic. In Figures(b) and (e), the top graph corresponds to the left Figures(a/d), and the bottom graph corresponds to right Figures(c/f).}
    \label{fig:pivot_pruning}
\end{figure}

Next, we describe an algorithmic enhancement to MWRP-\astar. When computing the mTSP heuristic, we know that any set of pivots, where no two pivots share watchers, results in an admissible mTSP heuristic. MWRP-\astar greedily uses as many pivots as possible. However, consider the example shown in Figure~\ref{fig:pivot_pruning}(a). From agent $a$, the closest watcher of $p_1$ is $w_1$, 15 cells away, and the closest watcher of $p_2$ is $w_3$ 5 cells away. Additionally, the closest watchers between $p_1$ and $p_2$, $w_1$ and $w_2$, are 2 cells apart. This results in the graph $G_\mathit{DLS}$ shown in Figure~\ref{fig:pivot_pruning}(b) (top graph). The mTSP heuristic solved on $G_\mathit{DLS}$ results in the path $a \rightarrow p_2 \rightarrow p_1$, with a heuristic value of $5 + 2 = 7$. However, if we remove $p_2$ (Figure~\ref{fig:pivot_pruning}(c)), then the only edge is $a \rightarrow p_1$ (Figure~\ref{fig:pivot_pruning}(b), bottom graph), and solving mTSP gives a heuristic value of $15$. We call this behavior \textit{shortcut removal}, since $p_2$ provides a shortcut between $a$ and $p_1$. To maximize the admissible heuristic value, we greedily remove pivots that may provide shortcuts in the mTSP solution.

The shortcut that pivot $p_i$ provides between agent $a_k$ and pivot $p_j$ is $s(p_i, p_j, a_k) = e(a_k, p_j) - (e(a_k, p_i) + e(p_i, p_j))$, where $e$ is the distance between two vertices in $G_\mathit{DLS}$. We use a while loop to repeatedly prune out the pivot that provides the largest shortcut between agents and other pivots until no more pivots provide a positive shortcut. Pivot pruning has a runtime complexity of $O(MP^3)$ where $P$ is the number of pivots. Pivot pruning is run for every mTSP heuristic calculation prior to running the mTSP solver. Pivot pruning is not guaranteed to increase the heuristic value, and in some cases may even reduce the informativeness of the heuristic. An example is shown in Figure~\ref{fig:pivot_pruning}(d-f), where removing pivot $p_3$, which provides a shortcut from $a$ to $p_2$, worsens the overall mTSP heuristic from $11$ to $9$. However, our experiments show that it provides a significant speedup.

\subsection{Parallel Heuristic Calculation}

Despite using lazy heuristic evaluation for \astar expansions, the majority of the time overhead in MWRP-\astar comes from the mTSP heuristic calculation. Parallelization algorithms for lazy heuristics have been explored for large-scale multiprocessing \cite{mukherjee2022mplp}, however these methods require constantly running asynchronous processes. Instead, we use Batch \astar \cite{li2022optimal, agostinelli2019solving}, which is often used to parallelize the calculation of neural network heuristics. During each expansion cycle of \astar, we first check if the next node to be popped has had its mTSP heuristic computed. If not, we parallelize the computation by looping through the next $N$ nodes that are going to be popped from \textit{OPEN} and accumulating the nodes that have not yet had their mTSP heuristic computed into a \textit{batch}. We then compute the mTSP heuristic for the batch in parallel and update each node's $f$-value in \textit{OPEN}.

Since $N$ is a constant, the runtime complexity of \astar is not changed. Larger $N$ values will parallelize more heuristic calculations, but may execute more unnecessary computations, since there is no guarantee that all of these nodes would have been expanded. We use $N = 100$ in our experiments.

As proven above, \algname is guaranteed to return an optimal solution. Notably, while our focus is the makespan objective, \algname makes no assumptions on the objective function and is thus optimal for graph-based WRP and MWRP with either the sum-of-costs or makespan objective.

\section{Bounded Suboptimal Methods}
\label{sec:suboptimal}
Bounded Suboptimal Search (BSS) algorithms are a class of search methods that, given a weight $w \ge 1$, guarantee a solution with cost $C \leq C^* \cdot w$, where $C^*$ is the cost of the optimal solution. We present two BSS algorithms that are used to scale \algname to larger maps and agent counts.

\subsection{Minimax Weighted \astar (\subalgname)}

A popular BSS variant of \astar is weighted \astar (W\astar) \cite{pohl1970heuristic}. In W\astar, the $f$-value is altered to $f_w(n) = g(n) + w \cdot h(n)$, where $g(n)$ is the cost to reach node $n$ and $h(n)$ is an admissible heuristic. W\astar has previously been applied to the graph-based WRP problem with WRP-\astar \cite{yaffe2021suboptimally}. To our knowledge, W\astar has not been applied to makespan problems with joint-space search.

\captionsetup[subfigure]{skip=5pt}   

\begin{figure}[t]
    \centering

    \includegraphics[width=\linewidth]{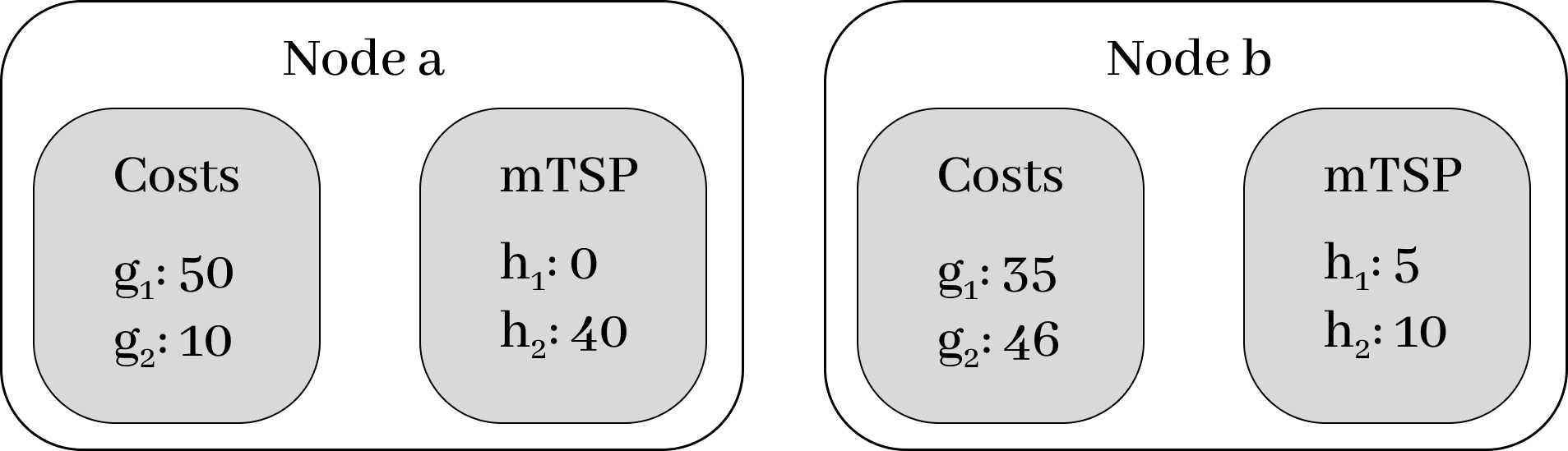}

    \caption{Two example search nodes, each with the agent costs ($g_k$) and estimates on the remaining search efforts ($h_k$).}
    \label{fig:mwastar_example}
\end{figure}

We propose \textit{Minimax Weighted \astar} (\subalgname), a W\astar variant specifically for the makespan objective. The goal is to prioritize agents with lower search effort while ensuring that the solution found is still bounded suboptimal. \subalgname first obtains a weighted $f_k$-value for each agent $a_k$. Then, the maximum of these $f_k$-values is used as the $w$-admissible $f$-value during the search. That is, \subalgname performs the \astar search with the f-value $f_{MxW}(n) = \max\limits_{k}\{g_k + w \cdot h_k\}$.

For example, consider nodes $a$ and $b$ shown in Figure~\ref{fig:mwastar_example}. Node $a$ has $g(a) = \max\limits\{g_1, g_2\} = 50$ and admissible $f(a) = \max\limits\{g_1 + h_1, g_2 + h_2\} = 50$. Similarly, node $b$ has $g(b) = \max\limits\{g_1, g_2\} = 45$ and admissible $f(b) = \max\limits\{g_1 + h_1, g_2 + h_2\} = 55$. In \astar, node $a$ is expanded first since it has a lower $f$-value.

However, using \subalgname, we generate $f_\mathit{MxW}(a) = \max\{50 + w \cdot 0, 10 + w \cdot 40\} = 10 + w \cdot 40$ and $f_\mathit{MxW}(b) = \max\{35 + w \cdot 5, 45 + w \cdot 10\} = 45 + w \cdot 10$. Thus, with $w \geq 1.2$, node $b$ would be expanded first, which is preferred since it has lower $h$-values and thus has a lower remaining search effort. We now prove that \subalgname is $w$-suboptimal.

\begin{theorem}
Given agent costs $g_k(n)$ and estimates of agent cost-to-go of $h_k(n)$, performing \astar search with $f_\mathit{MxW}(n) = \max\limits_i\{g_k() + w \cdot h_k(n)\}$ returns a solution with cost $C \leq w \cdot C^*$, where $C^*$ is the cost of the optimal solution.
\end{theorem}

\begin{proof}
We define the last node expanded during the search as node $G$ with solution cost $C$. Let $P^*$ be the optimal search path with cost $C^*$, and let $n$ be the last node on $P^*$ that is in \textit{OPEN}. Since $G$ was expanded before $n$, we know that $f_\mathit{MxW}(G) \leq f_\mathit{MxW}(n)$, and since $h_k(G) = 0$ for each agent $a_k$, $f_\mathit{MxW}(G) = g(G)$, so $g(G) \leq f_\mathit{MxW}(n) = \max\limits_k\{g_k(n) + w \cdot h_k(n)\} \leq \max\limits_k\{w \cdot g_k(n) + w \cdot h_k(n)\} = w \cdot f(n)$. This means $g(G) \leq w \cdot f(n)$. Since $f(n)$ is admissible and $n$ is on $P^*$, we know $f(n) \leq C^*$, meaning $g(G) = C \leq w \cdot C^*$.
\end{proof}

In each node, $g_k$ is the cumulative cost of agent $a_k$. For the Singleton heuristic, $h_{s,k}$ is the distance from agent $a_k$ to the nearest watcher of $s$. Using \subalgname, $f_{s,MxW} = \min\limits_k\{g_k + w \cdot h_{s,k}\}$ is a $w$-admissible $f$-value, so $f_{MxW} = \max\limits_s\{f_{s,MxW}\}$. For the mTSP $f$-value, $h_k$ is agent $a_k$'s path length along $G_\mathit{DLS}$, so we compute $f_\mathit{MxW}$ by minimizing $\max\limits_k\{g_k + w \cdot h_k\}$ as the mTSP objective function.

In the single-agent case, \subalgname is identical to W\astar. Additionally, following techniques from Anytime Weighted \astar \cite{hansen2007anytime}, we extend \subalgname to Anytime \subalgname (\anytimesubalgname). \anytimesubalgname continues to search for better solutions even after an initial solution was found. If $B$ is the cost of the best solution found by the search thus far, then \anytimesubalgname prunes nodes $n$ with $\max\{\frac{f_\mathit{MxW}(n)}{w}, g(n)\} \geq B$, since $\frac{f_\mathit{MxW}(n)}{w}$ is guaranteed to be admissible. This forces future solutions to improve on the current best solution.

\subsection{Focal Search}

The second BSS method utilizes \textit{Focal Search} (FS), which uses both an admissible heuristic $h(n)$ and a second, not necessarily admissible, heuristic $h_\mathit{FOCAL}(n)$ that serves as a better estimate of search effort \cite{pearl1982studies}. FS maintains the \textit{OPEN} list from \astar as well as a separate \textit{FOCAL} list. \textit{FOCAL} contains the subset of nodes $n$ in \textit{OPEN} with $f(n) \leq w \cdot f_{min}$, where $f_{min}$ is the current minimum $f$-value in \textit{OPEN}. Each iteration, the node with the lowest $h_\mathit{FOCAL}(n)$ value from \textit{FOCAL} is expanded.

Since we use lazy heuristic evaluation in \textit{OPEN}, we need to recompute the heuristic of some nodes before inserting them into \textit{FOCAL}. Each \astar expansion cycle, all nodes $n$ in \textit{OPEN} that were inserted with the Singleton $f$-value $f(n) \leq w \cdot f_{min}$ are popped and reinserted into \textit{OPEN} with the mTSP $f$-value. \textit{FOCAL} contains the subset of nodes $n$ in \textit{OPEN} with an mTSP $f$-value of $f(n) \leq w \cdot f_{min}$. We propose two different $h_\mathit{FOCAL}$ heuristics, both of which estimate the remaining search effort based on the paths produced by mTSP. The \textit{Sum Of Remaining Costs} (SORC) heuristic calculates $\sum\limits_k h_k$, and the \textit{Max Of Remaining Costs} (MORC) heuristic calculates $\max\limits_k h_k$.

We also utilize these heuristics with Anytime Focal Search (AFS) \cite{cohen2018anytime}. Defining $B$ as the cost of the best solution found thus far, AFS limits the nodes in \textit{FOCAL} to nodes with $f(n) < \min\{B, w \cdot f_{min}\}$ to force future solutions to improve on the current best solution.

\section{Postprocessing Framework}
\label{sec:postprocessing}
Postprocessing is utilized in planning algorithms to quickly improve the solution quality of an existing suboptimal solution. We propose a postprocessing framework that decomposes the problem into several single-agent searches that are each run on a separate subset of the map. Naturally, to improve the makespan of an existing suboptimal solution, $\paths$, we look to improve the path of the agent with the highest path cost while ensuring that the resulting set of paths is still a solution. Defining agent $a_{max}$ as the agent with the largest path cost, we want to replace $\pi_a$, the path of $a_{max}$, with $\pi_a'$ where (1) $c(\pi_a') < c(\pi_a)$ and (2) the resulting set of paths still has line-of-sight to every cell in $\unseen$.

To ensure that the new set of paths has line-of-sight to every cell in $\unseen$, we need to ensure that $\pi_a'$ sees all cells in $\unseen$ that no other agent sees. Thus, defining $\paths^-$ as the paths of all agents other than $a_{max}$, we extract the \textit{minimum responsibility} ($r$) of $a_{max}$ as $r = \unseen \setminus \bigcup\limits_{s \in \pathunion(\paths^-)} \los(s)$. Here, $r$ represents the cells that only $a_{max}$ sees in the original solution, and thus must see in any improved solution. We utilize \algname to optimally solve a single-agent subproblem for agent $a_{max}$ with $\unseen = r$ (Figure~\ref{fig:postprocessing}). Formally, we solve the problem tuple $(1, \freecells, r, \{\starts_{a_{max}}\}, \neighbors, \los, \objective)$.
\begin{figure}[t]
    \centering

    \includegraphics[width=\linewidth]{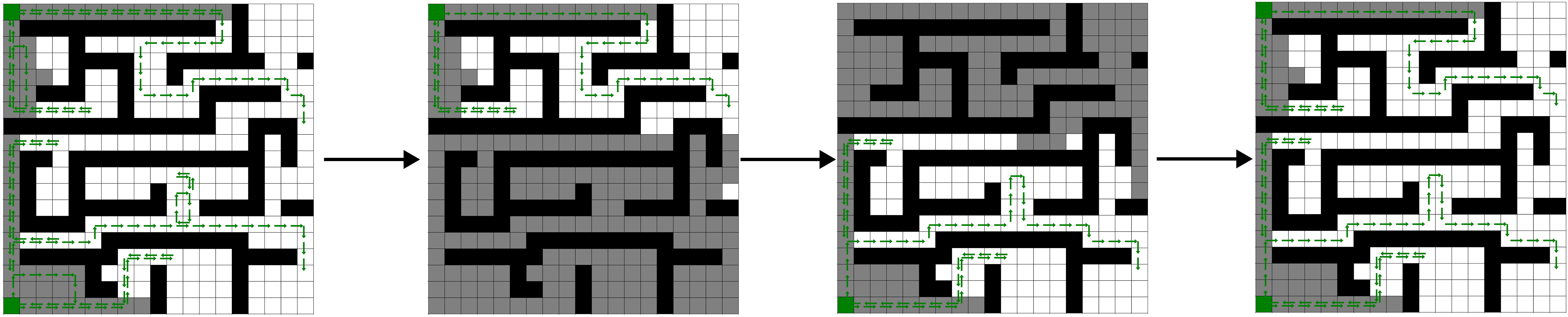}

    \caption{Visualization of the initial joint-space solution, decomposed agent sub-problems, and postprocessed solution.}
    \label{fig:postprocessing}
\end{figure}

Once a new path $\pi_a'$ is computed, we update $\pi_a = \pi_a'$. We repeat this process until the agent with the largest cost among $\paths$ has already had its path optimized in this decomposed manner (Figure~\ref{fig:postprocessing}). Since the agent responsibilities ($r$) may not be optimally partitioned, the postprocessing framework cannot guarantee a globally optimal solution. 

\section{Experimental Results}
\label{sec:experiments}
\begin{figure}[t]
    \centering

    \includegraphics[width=\linewidth]{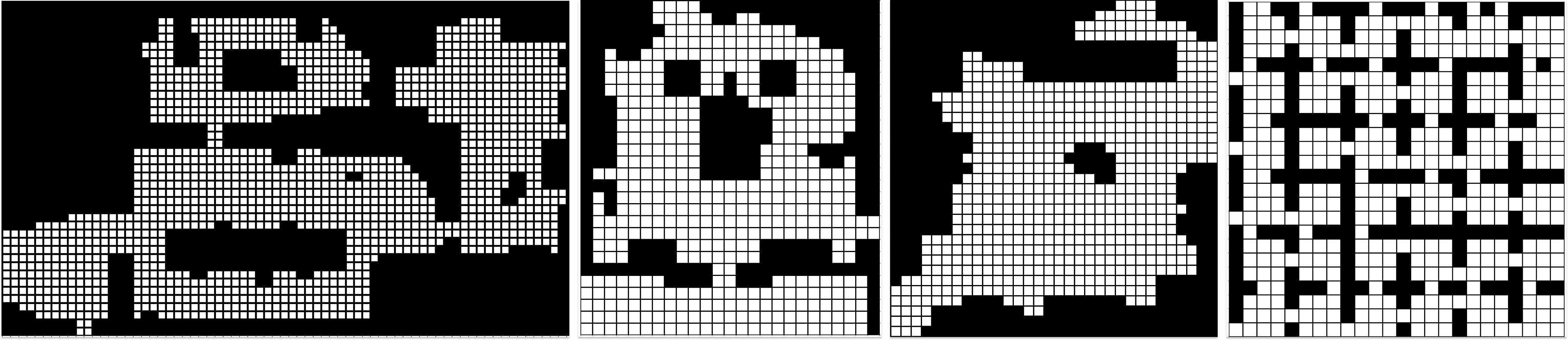}

    \caption{From left to right, Den101d, Lak105d, Den202d \cite{stern2019mapf}, and an example Room map.}
    \label{fig:experiment_maps}
\end{figure}

We conduct experiments with varying map types, map sizes, and agent counts to measure the performance of \algname, our BSS algorithms, and our postprocessing framework. For all of our experiments, we use four-way movement as the neighbor function and BresenhamLOS as the LOS function (Figure~\ref{fig:example_problem}). All algorithms, including the original MWRP-\astar algorithm, are implemented in C++. The mTSP heuristic is implemented as an integer linear programming formulation \cite{bektas2006multiple} using CPLEX \cite{cplex}. All experiments were run on a Legion 5i Gen 10 with 32 Intel Core Ultra 255Hx processors (4.50 GHz) and 32GB of RAM. We ran experiments on four different map styles (Figures ~\ref{fig:cell_path_domination_comparison} and ~\ref{fig:experiment_maps}): (1) Maze maps, (2) Room maps consisting of $3\times3$ rooms, (3) Random maps with 20\% obstacles, and (4) Game-inspired maps (Figures~\ref{fig:cell_path_domination_comparison} and \ref{fig:experiment_maps}). Room and random maps were custom-generated; the others were taken from the Moving AI MAPF Benchmark \cite{stern2019mapf}, previous work on MWRP \cite{livne2023optimally}, and previous work on multi-agent planning and task assignment \cite{chong2024optimal}. For all experiments, to simulate real-world scenarios where exploration requires entering an environment from the outside, we generate agent starting points by selecting random cells along the map border.

\subsection{State Space Reduction}

First, we compare the state space reduction and execution runtime of CD, PD, and CPD on three distinct map architectures (visualized in Figure~\ref{fig:cell_path_domination_comparison}). Averages were computed over 50 runs per agent count, from 1-5 agents, on each map. As shown in Table~\ref{tab:dominance_comparison}, PD and CPD were able to reduce $|\unseen|$ by more than 95\% on some maps. CD and PD both pruned more cells in structured maps with long corridors, which are less prominent in \textit{Random} maps and more prominent in \textit{Maze} maps. CPD executed 1.25$\times$ faster than PD.

\begin{table}[t]
    \centering
    \resizebox{\linewidth}{!}{\renewcommand{\arraystretch}{1.0} \begin{tabular}{ |>{\centering\arraybackslash}m{0.7cm}|p{2.0cm}|p{1.7cm}|p{1.7cm}|p{1.7cm}|  }
     \hline
     \multicolumn{2}{|c|}{\textbf{Map}}  & \centering Random & \centering Maze & \multicolumn{1}{|c|}{Minecraft} \\
     \multicolumn{2}{|c|}{\textbf{Average initial size of} $\mathbf{|\unseen|}$} & \centering $|\unseen| = $ 175.3 & \centering $|\unseen| = $ 598.9 & \multicolumn{1}{|c|}{$|\unseen| = $ 1279.5} \\

     \hline
    \multirow{3}{*}{\raisebox{\height}{CD}} & \textbf{\% Reduction} & \raggedleft\arraybackslash 0.8 $\pm$ 0.6 & \raggedleft\arraybackslash 69.9 $\pm$ 1.2 & \raggedleft\arraybackslash 18.5 $\pm$ 2.3 \\
                          & \textbf{Runtime} & \raggedleft\arraybackslash 2 $\pm$ 0 & \raggedleft\arraybackslash 4 $\pm$ 0 & \raggedleft\arraybackslash 45 $\pm$ 1 \\
     \hline
    \multirow{3}{*}{\raisebox{\height}{PD}} & \textbf{\% Reduction} & \raggedleft\arraybackslash \textbf{56.6 $\pm$ 19.8} & \raggedleft\arraybackslash \textbf{95.3 $\pm$ 1.1} & \raggedleft\arraybackslash \textbf{89.0 $\pm$ 3.8} \\
                          & \textbf{Runtime} & \raggedleft\arraybackslash 8 $\pm$ 0 & \raggedleft\arraybackslash 28 $\pm$ 1 & \raggedleft\arraybackslash 205 $\pm$ 3  \\
     \hline
    \multirow{3}{*}{\raisebox{\height}{CPD}} & \textbf{\% Reduction} & \raggedleft\arraybackslash \textbf{56.6 $\pm$ 19.8} & \raggedleft\arraybackslash \textbf{95.3 $\pm$ 1.1} & \raggedleft\arraybackslash \textbf{89.0 $\pm$ 3.8} \\
                          & \textbf{Runtime} & \raggedleft\arraybackslash \textbf{6 $\pm$ 1} & \raggedleft\arraybackslash \textbf{8 $\pm$ 0} & \raggedleft\arraybackslash \textbf{160 $\pm$ 2} \\
     \hline
    \end{tabular}}
    \caption{Comparison of CD, PD, and CPD on three distinct maps (shown in Figure~\ref{fig:cell_path_domination_comparison}). The \% reduction in the size of $\unseen$, as well as the CPU runtime (ms), are measured. The initial $\unseen$ is $C$ minus the initial LOS of each agent.}
    \label{tab:dominance_comparison}
\end{table}

\subsection{Search Algorithm Runtime Comparison}

\begin{figure*}[t]  
    \centering
    \includegraphics[width=0.95\textwidth]{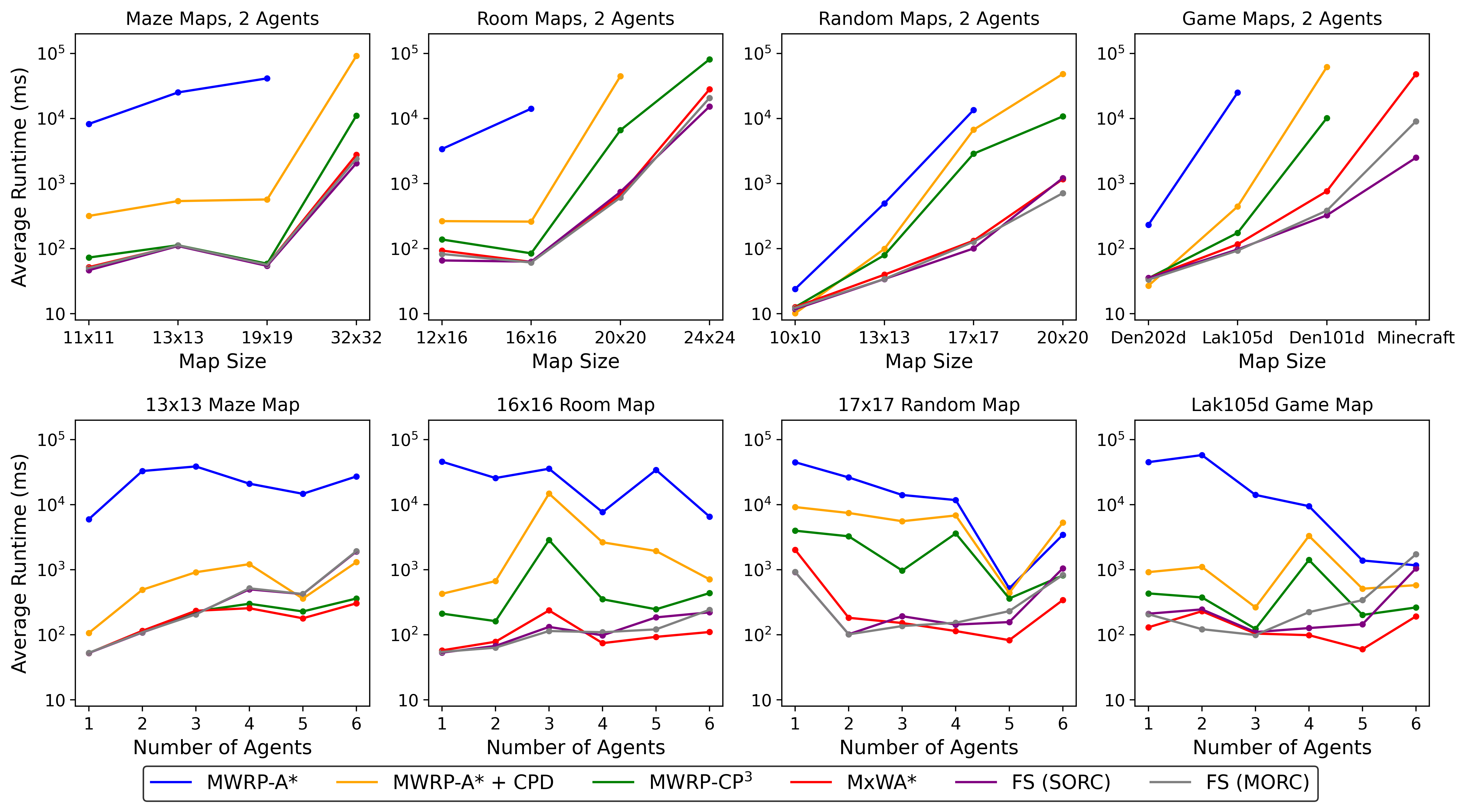}
    \caption{Runtime comparison of search algorithms with respect to map size (top) and agent counts (bottom). Suboptimal methods used $w = 2$. Algorithms were subject to a 200s time limit, and only those solving 80\%+ of the instances are shown.}
    \label{fig:main_experiments}
\end{figure*}

To our knowledge, MWRP-\astar is the only existing optimal algorithm for solving MWRP, and there are no bounded suboptimal algorithms. We compare runtime performance between MWRP-\astar, MWRP-\astar with CPD, \algname, \subalgname, FS with SORC, and FS with MORC. The runtime measurements are averaged over 20 instances, each with randomly generated agent starting locations. CPD runtime was included in the total runtime. As shown in Figure~\ref{fig:main_experiments}, \algname is the fastest optimal solver, running 200$\times$ faster than MWRP-\astar, and solving complex problem instances that MWRP-\astar failed to solve within a 200s time limit. The runtime improvement is significantly lower on \textit{Random} maps, as CPD is less effective (Table~\ref{tab:dominance_comparison}). Figure~\ref{fig:main_experiments} also demonstrates the scalability of our suboptimal solvers with respect to increasing map sizes and agent counts. For larger agent counts, \subalgname performs better than both FS algorithms, likely because the size of \textit{FOCAL} increases with a larger joint-space, resulting in more mTSP heuristic calculations.

\subsection{Ablation Study}

\begin{table}[t]
    \centering
    \resizebox{\linewidth}{!}{\begin{tabular}{ | >{\centering\arraybackslash}m{0.6cm}
                     | >{\centering\arraybackslash}m{1.7cm}
                     | >{\centering\arraybackslash}m{1.7cm}
                     | >{\centering\arraybackslash}m{1.7cm}
                     | >{\centering\arraybackslash}m{1.7cm}|  }
     \hline

     & \multicolumn{2}{|c|}{\textbf{32$\mathbf{\times}$32 Maze}} & \multicolumn{2}{|c|}{\textbf{Den101d}} \\
\cline{2-5}     
                          & $M=1$ & $M=3$ & $M=1$ & $M=3$ \\     
     \hline
    PP & 3.21 $\pm$ 1.26 & 8.18 $\pm$ 6.35 & 1.38 $\pm$ 0.19 & 4.99 $\pm$ 2.15 \\
     \hline
    PHC & 2.03 $\pm$ 0.13 & 5.59 $\pm$ 2.00 & 1.76 $\pm$ 0.30 & 2.33 $\pm$ 0.83 \\
     \hline
    \end{tabular}}
    \caption{Ablation tests with $M$ agents measuring the scalar factor increase in runtime after removing either PP or PHC.}
    \label{tab:individual_improvement}
\end{table}

We also performed an ablation study to measure the individual improvement of the Pivot Pruning (PP) and Parallel Heuristic Calculation (PHC) techniques. Table~\ref{tab:individual_improvement} shows the slowdown caused by removing PP or PHC, each averaged over 20 different problem instances. Both methods provide a speedup in all test cases, and both provide a larger benefit on problems with more agents. PHC has low variances in its slowdown factor, indicating that its speedup is consistent regardless of the problem. PP, however, has high variances, meaning its speedup varies from problem to problem.

\subsection{Comparison of Suboptimal Variants}

\begin{figure}[t]
    \centering

    \begin{subfigure}{0.50\linewidth}
        \centering
        \includegraphics[width=\linewidth]{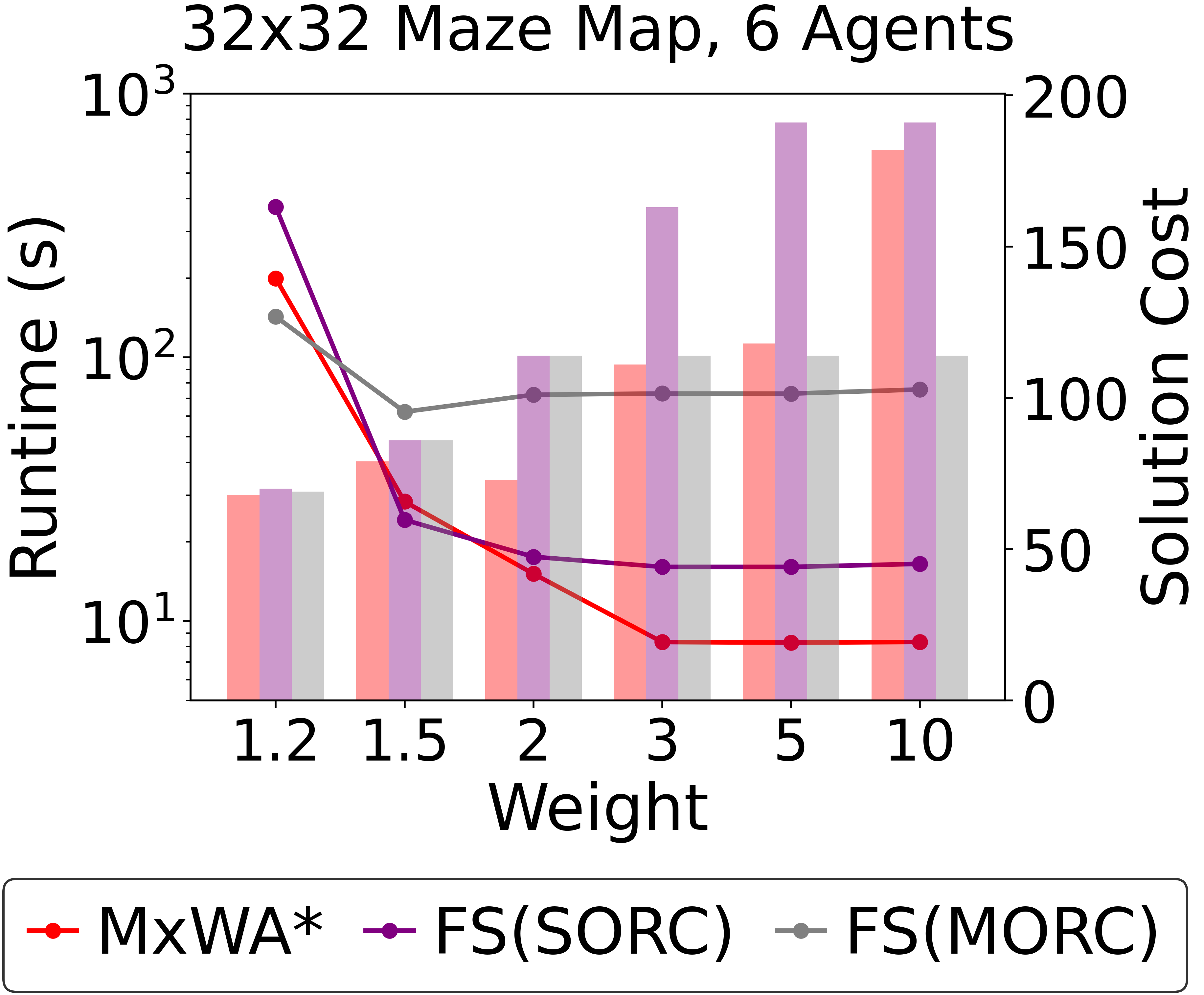}
        \caption{}
    \end{subfigure}
    \hfill
    \begin{subfigure}{0.46\linewidth}
        \centering
        \includegraphics[width=\linewidth]{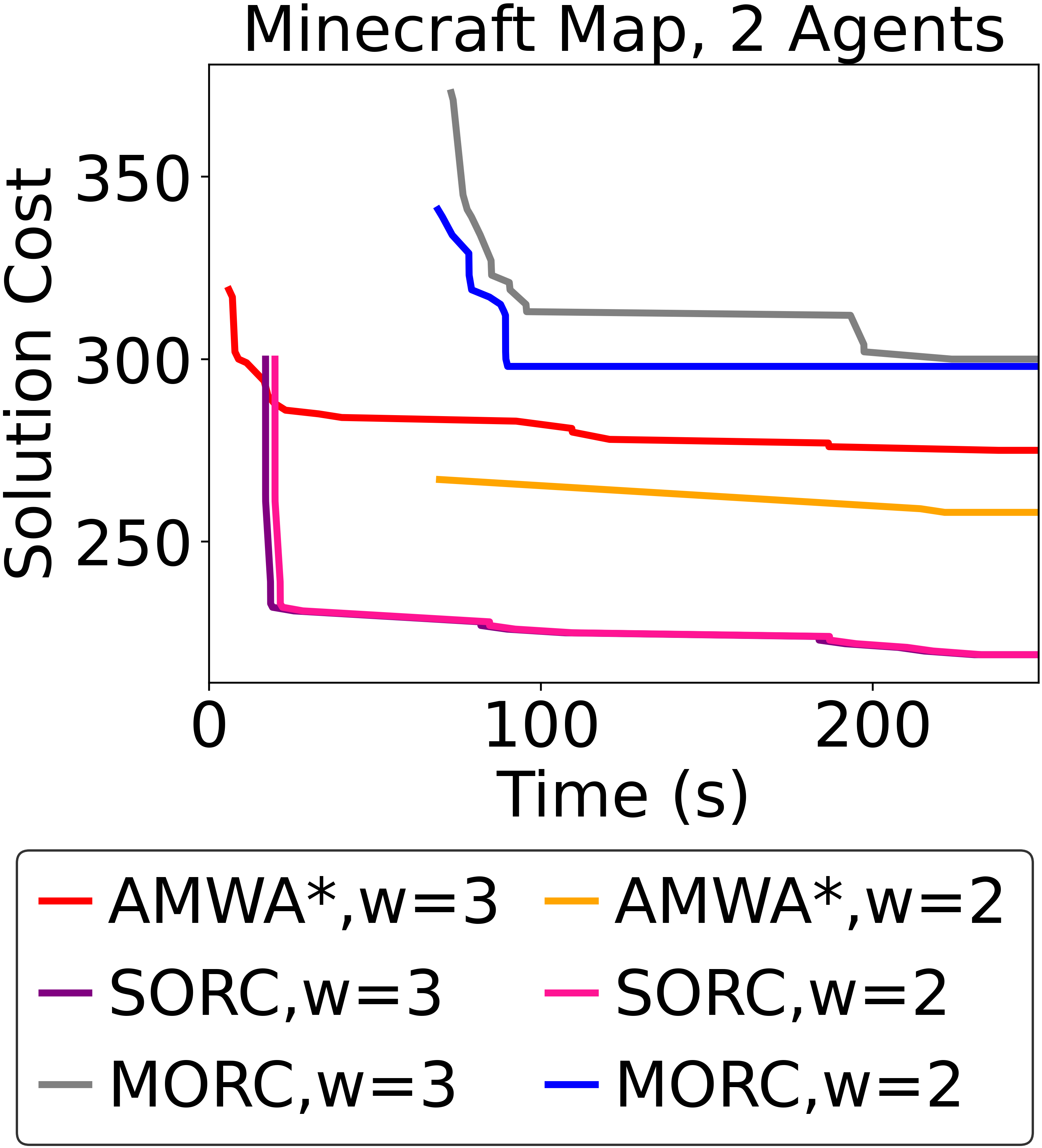}
        \caption{}
    \end{subfigure}

    \caption{Figure(a) shows runtime (lines) and solution cost (bars) at different $w$ values. Figure(b) shows the cost improvement over time of the anytime variations.}
    \label{fig: w_comparison_and_anytime_results}
\end{figure}

As shown in Figure~\ref{fig: w_comparison_and_anytime_results}(a), we measured the performance of all three suboptimal methods at different weight values by running the algorithms on the same problem instance, a 32$\times$32 Maze Map with 6 agents, at varying $w$ values. The runtime and solution cost were averaged over 10 iterations. All three algorithms have increasing solution costs increasing and decreasing runtimes with larger $w$ values. At lower $w$ values, all three algorithms perform similarly in terms of runtime and solution cost. At higher $w$ values, \subalgname is able to achieve faster runtimes than either FS.

We also empirically tested the anytime behavior of \anytimesubalgname, AFS with SORC heuristic, and AFS with MORC heuristic, as shown in Figure~\ref{fig: w_comparison_and_anytime_results}(b). All three algorithms were run on the Minecraft-inspired map with 2 agents. The algorithms were run only once, since the improvement over time cannot be averaged over several runs. We can see that all three algorithms improve their costs significantly over the course of the runtime. AFS with the SORC heuristic takes longer than \anytimesubalgname to generate an initial solution (for $w = 3$), but is able to quickly improve its solution much faster than \anytimesubalgname and AFS with the MORC heuristic. These results demonstrate that all three algorithms are effective depending on the specific scenario and $w$ value.

\subsection{Impact of Postprocessing Framework}

\begin{table}[t]
    \centering
    \resizebox{\linewidth}{!}{\begin{tabular}{ | >{\centering\arraybackslash}m{1.8cm}
                     | >{\centering\arraybackslash}m{0.5cm}
                     | >{\centering\arraybackslash}m{1.25cm}
                     | >{\centering\arraybackslash}m{1cm}
                     | >{\centering\arraybackslash}m{1.25cm}
                     | >{\centering\arraybackslash}m{1cm}|  }
     \hline

    \multirow{2}{*}{} & \multirow{2}{*}{\textbf{w}} & \multicolumn{2}{|c|}{\textbf{BSS Algorithm}} & \multicolumn{2}{|c|}{\textbf{Postprocessing}} \\
\cline{3-6}
                          &  & \textbf{Runtime} & \textbf{Cost} & \textbf{Runtime} & \textbf{Cost} \\     
     \hline
    \algname & 1 & 7,379 & 92 & 0 & 92 \\
     \hline
    \multirow{2}{=}{\subalgname} & 2 & 1,277 & \textbf{124} & 164 & \textbf{93} \\
                          & 5 & 1,154 & 129 & 146 & 109 \\
     \hline
    \multirow{2}{=}{FS (SORC)} & 2 & 1,121 & 162 & 168 & 125 \\
                          & 5 & 1,105 & 162 & 182 & 125 \\
     \hline
    \multirow{2}{=}{FS (MORC)} & 2 & \textbf{1,074} & 127 & \textbf{156} & 107 \\
                          & 5 & 1,091 & 127 & 166 & 107 \\
     \hline
    \end{tabular}}
    \caption{Impact of postprocessing architecture on BSS algorithms. The BSS runtime (ms), cost of the initial BSS solution, postprocessing runtime (ms), and improved solution cost are shown. We also show \algname as a reference.}
    \label{tab:bss_results}
\end{table}

The postprocessing framework was tested on a 32$\times$32 Maze with 3 agents. As shown in Table~\ref{tab:bss_results}, for each algorithm, the postprocessing framework increased overall runtime by only 8\% while improving the solution cost significantly, with \subalgname achieving a near-optimal postprocessed solution. In general, the postprocessing framework is best suited for \subalgname since its solution better distributes the paths throughout the map, which helps create even partitions.

\section{Conclusion}
\label{sec:discussion}
In this work, we introduced \algname, which optimally solves MWRP by utilizing cell and path dominance, pivot pruning, and parallel heuristic computation, and computes optimal paths 200$\times$ faster than existing baselines. We also introduced \subalgname, a variation of W\astar for the makespan objective function, and SORC and MORC, two heuristics used in conjunction with Focal Search. Additionally, we presented anytime variations for our BSS algorithms, as well as a postprocessing framework, both of which significantly improve the quality of the suboptimal solution. Future work should (1) investigate faster heuristics \cite{ren2024bounded}, (2) experiment on non-grid graphs \cite{prms}, and (3) expand the problem scope (i.e. considering inter-agent collisions or heterogeneous agents).


\section*{Acknowledgments}
\label{sec:acknowledgements}
The research at Carnegie Mellon University was partially supported by the National Science Foundation under grants \#2328671 and \#2441629.
Ariel Felner was supported by Israel Science Foundation (ISF) Grant \#909/23 and by a grant from the Israeli Ministry of Science and Technology (MOST).

\bibliography{refs}

\end{document}